\documentclass[3p]{elsarticle}
\usepackage[utf8]{inputenc}
\usepackage{amsmath}
\usepackage{amsfonts}
\usepackage{amssymb}
\usepackage{amsfonts}
\usepackage{amsthm}
\usepackage{amscd}
\usepackage{color}
\usepackage{graphics}
\usepackage{graphicx}
\usepackage{psfrag}
\usepackage[toc,page]{appendix}

\usepackage[all]{xy}

\usepackage{algorithm,algorithmic}

\usepackage{tikz}
\usetikzlibrary{shapes.geometric, arrows}

\tikzstyle{start} = [rectangle, rounded corners, minimum width=3cm, minimum height=1cm,text centered, draw=black]
\tikzstyle{stopno} = [rectangle, rounded corners, minimum width=3cm, minimum height=1cm,text centered, draw=red]
\tikzstyle{stopyes} = [rectangle, rounded corners, minimum width=3cm, minimum height=1cm,text centered, draw=green]

\tikzstyle{io} = [trapezium, trapezium left angle=70, trapezium right angle=110, minimum width=3cm, minimum height=1cm, text centered, draw=black]
\tikzstyle{process} = [rectangle, minimum width=3cm, minimum height=1cm, text centered, draw=black]
\tikzstyle{decision} = [diamond, minimum width=3cm, minimum height=1cm, aspect = 2,draw=black]

\tikzstyle{arrow} = [thick,->,>=stealth]

\usepackage{natbib}
\biboptions{sort&compress,numbers}

\newtheorem{thm}{\bf Theorem}[section]
\newtheorem{lem}[thm]{\bf Lemma}
\newtheorem{prp}[thm]{\bf Proposition}

\newtheorem{exmp}[thm]{\bf Example}
\newtheorem{cor}[thm]{\bf Corollary}

\newcommand{\cis}[1][C]{\ensuremath{\mathbb{#1}}}
\newcommand{\av}[1]{\ensuremath{\mathcal{#1}}}

\newcommand{\vek}[1][h]{\ensuremath{\mathbf{#1}}}
\newcommand{\f}[1]{\mathbf{#1}}

\newcommand{\pr}[1]{\ensuremath{\mathbb{P}^{#1}_{\cis}}}

\newcommand{\euR}[1]{\ensuremath{\mathbb{E}^{#1}_{\cis[R]}}}

\newcommand{\genus}{\ensuremath{\mathrm{g}}}
\newcommand{\pic}{\ensuremath{\mathrm{Cl}}}
\newcommand{\cyc}{\ensuremath{\mathrm{c}}}

\journal{arXiv}

\begin{document}

\sloppy

\begin{frontmatter}

\title{Recognizing implicitly given rational canal surfaces}

\author[plzen1]{Jan Vr\v{s}ek}
\ead{vrsek@ntis.zcu.cz}

\author[plzen2,plzen1]{Miroslav L\'avi\v{c}ka\corref{cor1}}
\cortext[cor1]{Corresponding author}
\ead{lavicka@kma.zcu.cz}

\address[plzen1]{NTIS -- New Technologies for the Information Society, Faculty of Applied Sciences, University of West Bohemia,
         Univerzitn\'i 8, 301 00 Plze\v{n}, Czech~Republic}

\address[plzen2]{Department of Mathematics, Faculty of Applied Sciences, University of West Bohemia,
         Univerzitn\'i~8,~301~00~Plze\v{n},~Czech~Republic}

\begin{abstract}
It is still a challenging task of today to recognize the type of a given algebraic surface which is described only by its implicit representation. In~this paper we will investigate in more detail the case of canal surfaces that are often used in geometric modelling, Computer-Aided Design and technical practice (e.g. as blending surfaces smoothly joining two parts with circular ends). It is known that if the squared medial axis transform is a rational curve then so is also the corresponding surface. However, starting from a polynomial it is not known how to decide if the corresponding algebraic surface is rational canal surface or not. Our goal is to formulate a simple and efficient algorithm whose input is a~polynomial with the coefficients from some subfield of $\cis[R]$ and the output is the answer whether the surface is a rational canal surface. In the affirmative case we also compute a rational parameterization of the squared medial axis transform which can be then used for finding a rational parameterization of the implicitly given canal surface.
\end{abstract}


\end{frontmatter}

%%%%%%%%%%%%%%%%%%%%%%%%%%%%%%%%%%%%%%%%%%%%%%%%%%%%%%%%%%%%%%%%%%%%%%%%%%%%%%%%%%%%%%%%%%%%%%%%%%%%%%%%%%%%%%%%%%%%%%%%%%%%%%%%%%%%%%%
\section{Introduction and related work}\label{sec intro}
%%%%%%%%%%%%%%%%%%%%%%%%%%%%%%%%%%%%%%%%%%%%%%%%%%%%%%%%%%%%%%%%%%%%%%%%%%%%%%%%%%%%%%%%%%%%%%%%%%%%%%%%%%%%%%%%%%%%%%%%%%%%%%%%%%%%%%%

In this paper we will pay an attention to the so called canal surfaces. These surfaces, which are defined as envelopes of moving spheres in 3-space, are very popular in Computer-Aided Design as they are often used as blending surfaces between the parts with circular ends. It was proved in \cite{PePo97} that any canal surface with a rational spine curve (a set of all centers of moving spheres) and a rational radius function possesses a rational parameterization. An algorithm for generating rational parameterizations of canal surfaces was developed and investigated in \cite{LaSchWi01}. The class of rational canal surfaces with a rational spine curve and  a rational radius function is a proper subset of the class of rational canal surfaces -- all canal surfaces with  a rational spine curve and rational squared radius function admit a rational parameterization, cf.  \cite{Pe98,BaJuLaSchSi14}.
The reverse problem (a rational parametric description of a curve or a surface is given, find the corresponding implicit equation) is called the implicitization problem. An algorithm for computing the implicit equation of a canal surface generated by a rational family of spheres is presented in \cite{DoZu09}.

However our goal is different  -- we start with an~implicit representation (i.e., with some polynomial in $x,y,z$) and want to decide if the corresponding algebraic surface is a rational canal surface or not. This is still a challenging problem, only partially solved in the recent past for a subfamily of canal surfaces, namely for the surfaces of revolution -- see \cite{VrLa14}. Moreover, in case of the positive answer we want to compute the equation of the spine curve and also the radius function. We would like to emphasize that this study is interesting not only from the theoretical point of view but it also reflects a need of the real-world applications as the results of many geometric operations are often described only implicitly. Then it is a challenging task to recognize the type of the obtained surface, find its characteristics and for the rational surfaces compute also their parameterizations. This is needed e.g. when the implicit blend surfaces (often of the canal-surface type) are constructed, see \cite{HoHo85,Ro89, Ha90,Ha01}.

\medskip
Now, we start with short recalling some elementary notions. Let $\euR{3}$ be Euclidean 3-space equipped with the Cartesian coordinates. A point $\vek[x]$ is represented
with respect to a coordinate system by a vector $(x, y, z)$, and we do not distinguish between the point and its coordinate vector.
Points in the projective closure of $\euR{3}$ will be described using standard homogeneous coordinates
\begin{equation}
(W:X:Y:Z) = (1:x:y:z).
\end{equation}
The equation $W=0$ describes the {\em ideal plane} as the set of all asymptotic directions, i.e.,
of points at infinity. The subset of the ideal plane which is invariant with respect to all similarities is called the {\em absolute conic section} $\Omega$ and characterized by
\begin{equation}
\Omega:\,\, X^2+Y^2+Z^2 = W=0,
\end{equation}
consisting of solely imaginary points. Circles are such conic sections which intersect $\Omega$ in two distinct points. Moreover,  circles lying in parallel planes possess the same points on $\Omega$.

\begin{figure}[t]
\begin{center}
   {\psfrag{X}{}
   \psfrag{Sigma}{$\Sigma(t)$}
   \psfrag{s}{$\f s(t)$}
   \psfrag{ds}{$\f s'(t)$}
   \psfrag{r}{$r(t)$}
   \psfrag{p}{$\f s$}
   \psfrag{Cp}{$\av{C}_{\f s}$}
   \includegraphics[width=0.4\textwidth]{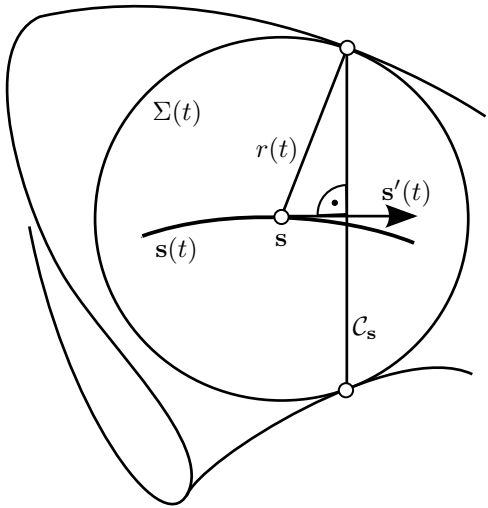}
   }
\begin{minipage}{0.9\textwidth}
\caption{A canal surface with the spine curve $\f s(t)$ and the radius function $r(t)$ as the envelope of spheres $\Sigma(t)$. $\av{C}_{\f s}$ is the characteristic
circle corresponding to the point $\f s$ on the spine curve. \label{canal}}
\end{minipage}
\end{center}
\end{figure}

A {\em canal surface} is defined as the envelope of a one-parameter family of spheres $\Sigma(t)$ whose centers trace a curve $\av{S}$ in 3-dimensional space parameterized by $\vek[s](t)$ and possess radii $r(t)$, i.e.,
\begin{equation}\label{oneparfamily}
\Sigma(t):\, |(x,y,z)^T-\vek[s](t)|^2-r(t)^2= 0.
\end{equation}
The curve $\av{S}$ is called the {\em spine curve} and $r(t)$ the {\em radius function} of the canal surface. For constant $r(t)$ we obtain a {\em pipe surface}, for  $\av{S}$ being a straight line we arrive at a {\em surface of revolution}.

By appending the corresponding sphere radii $r$ to the points of the spine curve (or the {\em skeleton}, or the {\em medial axis}) we obtain the {\em medial axis transform} (shortly  MAT), i.e., a curve $\av{M}$ in 4-dimensional space. In addition, we define the curve $\av{M}^2$ (and called a squared MAT) which is obtained by appending the squared sphere radii $R(t)$ (i.e., $r(t)=\sqrt{R(t)}$) to the points of the spine curve. The rationality of canal surfaces is described by the following proposition, cf. \cite{BaJuLaSchSi14}:

\begin{prp}
Any canal surface with the corresponding squared medial axis transform $\av{M}^2$ possesses a~rational parameterization if and only if $\av{M}^2$ is a rational curve.
\end{prp}

The implicit equation $f(x,y,z) =0$ of the canal surface given by \eqref{oneparfamily} can be obtained by eliminating the parameter $t$ from the defining equations
\begin{equation}\label{eq canal_eq}
\Sigma(t)=0,\quad \Sigma'(t)=0,
\end{equation}
where $\Sigma'$ is the derivative of $\Sigma$ with respect to $t$. The equation $\Sigma'(t)=0$ describes the plane with the normal vector ${\f s'(t)}$, i.e., perpendicular to the spine curve $\mathbf{s}(t)$. Thus the canal surface contains a one-parameter set of the so called {\em characteristic circles} $\av{C}_{\vek[s](t)}=\Sigma(t)\cap \Sigma'(t)$, which shows that the canal surfaces are generated by one-parameter family of circles.

%belong among the so called {\em ringed surfaces}.

We recall a fundamental property of canal surfaces
\begin{prp}
 Let  be given a canal surface  with the spine curve $\av{S}$. Then the normal lines through its non-singular points lying on the same characteristic circle $\av{C}_{\vek[s]}$ intersect the spine $\av{S}$ in the    point $\vek[s]$.
\end{prp}

This enables us to formulate a naive approach to the recognition of rational canal surfaces. It suffices to find  parameterizations of two rational curves corresponding in parameter (this means that for any choice of the parameter, the corresponding points lie on the same characteristic circle). Hence the intersection of the~normals yields the rational parameterization of the spine curve and thus also of the squared radius function. Nevertheless it must be stressed out that finding rational curves on an~algebraic surface is a hard problem and moreover computing their parameterizations corresponding in parameter is another very difficult task. Hence we will develop a~different method, based on square-root parameterizations of plane sections of a given surface.

Finally, throughout the paper we assume that a~given real algebraic surface $\av{X}_{\cis[R]}$ is defined by the~real polynomial $f(x,y,z)$. When homogenizing it we arrive at the polynomial $F(W,X,Y,Z)=0$ which describes the projective closure of the surface. In addition, we assume that the surface is considered in the~complex extension, i.e., $\av{X}\subset\pr{3}$.

%Let $\pi:\af{3}\times\af{1}\rightarrow\af{3}$ be the~projection forgetting the radius coordinate. Denote the image of $\av{M}^2$ under this mapping by $\av{S}$, i.e., $\av{S}$ is a~spine curve of $\av{X}$.

%%%%%%%%%%%%%%%%%%%%%%%%%%%%%%%%%%%%%%%%%%%%%%%%%%%%%%%%%%%%%%%%%%%%%%%%%%%%%%%%%%%%%%%%%%%%%%%%%%%%%%%%%%%%%%%%%%%%%%%%%%%%%%%%%%%%%%%
\section{Plane sections of rational canal surfaces}\label{sec preparation}
%%%%%%%%%%%%%%%%%%%%%%%%%%%%%%%%%%%%%%%%%%%%%%%%%%%%%%%%%%%%%%%%%%%%%%%%%%%%%%%%%%%%%%%%%%%%%%%%%%%%%%%%%%%%%%%%%%%%%%%%%%%%%%%%%%%%%%%

The simplest instances of canal surfaces are surfaces of revolution -- in this case the spine curve is just a straight line. An algorithm for recognizing surfaces of revolution and deciding on their rationality has been recently presented in \cite{VrLa14}. So, throughout this paper it is assumed that that $\av{X}$ is not a surface of revolution.

As recalled above, the canal surfaces contain a one-parameter family of characteristic circles. Each of these circles intersects the absolute conic section $\Omega$ in two points. As the spine curve is not a straight line, all these intersection points cover the whole absolute conic and thus $\Omega\subset\av{X}$. The multiplicity of $\Omega$ in $\av{X}$ is called the~\emph{cyclicity} of $\av{X}$ and denoted by $\cyc(\av{X})$.

\begin{prp}\label{prp non-zero cyclicity}
  Canal surface $\av{X}$ which is not a~surface of revolution has $c(\av{X})>0$.
\end{prp}

This simple criterion eliminates most of the candidates on canal surfaces because a~generic surface has $\mathrm{c}(\av{X})=0$. Now the remaining part of the method will be devoted to the reconstruction of the curve $\av{M}^2$ from the given polynomial $f(x,y,z)$. If it is not possible to reconstruct $\av{M}^2$ (see the details in Section~\ref{sec reconstruction}), or the canal surface $\av{Y}$ corresponding to the computed $\av{M}^2$ differs from $\av{X}$ then we can state that the given $\av{X}$ is not a~rational canal surface. Otherwise the output of the algorithm will be positive and $\av{M}^2$ may be then used to parameterize the surface, see \cite{LaSchWi01}.

\begin{lem}\label{lem hyperelliptic projection}
  Let $\av{X}$ be a rational canal surface with the spine curve $\av{S}$ and $\av{H}$ be a generic plane section. Then there exists a~rational mapping $\xi:\av{H\dashrightarrow S}$ such that its fiber at a generic point $\vek[s]\in\av{S}$ consists of the two points which are obtained as the intersection points of the plane containing $\av{H}$ and the characteristic circle $\av{C}_{\vek[s]}$. Such a mapping is called \emph{coherent}.
\end{lem}

\begin{proof}
  Let $(\vek[s](t),R(t))$ be a proper parameterization of $\av{M}^2$, the algorithms presented in   \cite{BaJuLaSchSi14,Pe98} allow to construct a proper parameterization $\vek[x](t,s)$ of the resulting canal surface (not necessarily with real coefficients) such that for a~fixed $t_0\in\cis$ the parametric curve $\vek[x](t_0,s)$ is a characteristic circle determined by the sphere $(\vek[x]-\vek[s](t_0))^2=R(t_0)$. The desired mapping $\xi$ is then the~restriction of the~rational map $\vek[s]\circ \vek[x]^{-1}:\av{X\dashrightarrow S}$ to the section $\av{H\subset X}$.
\end{proof}

Let $\av{H}$ be a generic irreducible plane section. Then the mapping $\xi$ from Lemma~\ref{lem hyperelliptic projection} is a double cover of a rational curve. As known $\av{H}$ has to be a rational, elliptic or hyperelliptic curve. These curves are birational to plane curves in the so called Weierstrass form $\av{E}:y^2-P(x)=0$, where $P(x)$ is a~square-free polynomial. The algorithms for finding the~Weierstrass form of elliptic and hyperelliptic curves, together with a~birational morphism $\sigma:\av{E\rightarrow H}$, were formulated in \cite{vH95} and \cite{vH02}, respectively.

A curve $\av{E}$ can be directly parameterized using square roots as $\vek[e]^\pm(t)=(t,\pm\sqrt{P(t)})$ and via $\sigma$ we arrive at the square root parameterization  $\vek[h]^\pm=\sigma\circ\vek[e]^\pm$ of $\av{H}$. Clearly, this is not a mapping as one parameter $t$ leads to two points $\sigma(t,\pm\sqrt{P(t)})$.

Two regular points on the canal surface are called \emph{corresponding} whenever they lie on the same characteristic circle. If $\sigma(t,\pm\sqrt{P})$ are corresponding points for a~generic $t$ then this square-root parameterization  may be used for computing the squared MAT. However this is not necessarily satisfied and thus we one has to be very careful when constructing required square-roots parameterizations.

Recall that any mapping $\av{H}\dashrightarrow\pr{1}$ of degree two, which
glues two corresponding points (i.e., points from $\av{H}\cap\av{C}_{\vek[s]}$), is called coherent. Such a mapping may be viewed as a kind of inverse to the desired square-root parameterization. Therefore in what follows we will go through all the three possible cases and we show how to construct this mapping.

\subsection*{Case I: $\genus(\av{H})=0$}

This case is exceptional as the only canal surfaces with rational generic sections can be surfaces of revolution which were excluded from further investigations at the beginning of the paper. The proof of this~statement in fact copies the classification of multiple conical surfaces with rational plane sections from \cite{Sch01}. The only difference is that in our situation we do not have to require that the surface contains more families of conics -- it is sufficient that it contains at least one family. Moreover these conics are circles in our case. So they intersect $\Omega$ in two points with the multiplicity $\cyc(\av{X})$. Exactly the same argument was used in \cite{Lu13} for the classification of the surfaces with more than one family of circles.

From these reasons, we do not go into details -- it is enough to recall only the notation and the results from the above mentioned papers and then to apply them to surfaces with one-parameter family of circles; for more details we refer to the original sources.

The linear normalization  $\overline{\av{X}}$ of the surface $\av{X}$ with rational plane sections is either a ruled surface $\av{R}_{n,m}\subset\pr{2m+n+1}$, or the Veronese surface $\av{V}\subset\pr{5}$. Recall that the Veronese surface is given by the parameterization
\begin{equation}
    (1:s:t:s^2:st:t^2)
\end{equation}
with the class group of divisors modulo linear equivalence $\pic(\av{V})$ generated by $L$ (the divisor $t=0$) with the intersection product $L^2=1$. The ruled surfaces $\av{R}_{n,m}$ are parameterized by
  \begin{equation}
    (1:t:\dots:t^{m+n}:s:st:\dots:st^m)
  \end{equation}
with the ruling $P$ and the cross section $B$ generating $\pic(\av{R}_{n,m})$ where $P^2=0$, $P\cdot B=1$ and $B^2=-n$.

\begin{thm}\label{thm rational section then sor}
  Rational canal surface $\av{X}$ with rational generic plane section is a surface of revolution.
\end{thm}

\begin{proof}

Let $H\in\pic(\overline{\av{X}})$ denote the divisor class of the pullback of hyperplane section $\av{H}$, $C$ the pullback of a generic circle and $\Omega$ the pullback of the absolute conic section.  Then $H\cdot C=2$ since the generic plane intersects the generic circle in two points. On $\av{V}$ we have $H\sim 2L$ and on $\av{R}_{n,m}$ is $H\sim B+(m+n)P$. With this it is possible to enumerate all admissible combinations  of $m$ and $n$. We pick up from \cite{Sch01}
only these cases where $\av{X}$ contains at least 1-parameter family of conics.

\begin{enumerate}
   \item $\overline{\av{X}}=\av{V}$\\
      It follows from \cite{Lu13} that there does not exist such a surface with the non-zero cyclicity.
  \item $\overline{\av{X}}=\av{R}_{0,2}$\\
    This is the only case which is not contained in \cite{Lu13} because it contains only one-parameter family of conics. We have $\deg(\av{X})=\deg(\av{R}_{0,2}) = 4$ and thus $0<\cyc(\av{X})\leq 2$. Nevertheless a quartic surface with the cyclicity two is a Darboux cyclide which possesses elliptic plane sections. Hence, assume $\cyc(\av{X})=1$ and therefore $\Omega\sim C$.  Moreover  $C\sim B$ and we arrive at
   \begin{equation}
    2=\Omega\cdot C = B^2=-1,
  \end{equation}
  which is a~contradiction.
  \item $\overline{\av{X}}=\av{R}_{0,1}$, with $C\sim 2B$ or $\overline{\av{X}}=\av{R}_{2,0}$\\
    We have $\deg\av{R}_{0,1}=\deg\av{R}_{2,0}=2$. However, the only quadric surface with the non-zero cyclicity is a sphere, which is again a surface of revolution.
  \item $\overline{\av{X}}=\av{R}_{1,1}$\\
    This is analogous to the case 1.
\end{enumerate}

To conclude, the only possibility for a surface with rational sections and families of circles is that it possesses the zero cyclicity; by Proposition~\ref{prp non-zero cyclicity} it must be a surface of revolution.
\end{proof}

\subsection*{Case II: $\genus(\av{H})=1$}
Firstly, we stress out that the class of these canal surfaces is not empty any more.
For example, it is well known that a generic plane section $\av{H}$ of a~Dupin cyclide is an~elliptic curve. Recall that in this case there exists a birational morphism $\sigma:\av{E\rightarrow H}$, where $\av{E}\subset\pr{2}$ is a non-singular cubic curve in the~Weierstrass form.

\begin{thm}\label{thm construction of coherent map}
  Assume that $\vek[p],\vek[q]\in\av{H}$ are two corresponding  points and denote by $\vek[r]$ the third intersection point of the cubic $\av{E}$ with the line connecting $\sigma^{-1}(\vek[p])$ and $\sigma^{-1}(\vek[q])$.  The projection from $\vek[r]$ defines a morphism $\pi:\av{E}\rightarrow\pr{1}$ such that the composed mapping $\pi\circ\sigma^{-1}:\av{H}\dashrightarrow\pr{1}$ is the desired coherent mapping.
\end{thm}

\begin{proof}
  We know that a~coherent mapping really can be found -- consider  e.g. $\zeta:\av{H\dashrightarrow S}\dashrightarrow\pr{1}$, where the second birational mapping exists because $\av{S}$ is rational. Let $\mu:\,\av{\hat{H}\rightarrow H}$ be the~desingularization of $\av{H}$ and  $P$, $Q$ be the divisor classes on $\hat{\av{H}}$ corresponding to the preimages of $\vek[p]$ and $\vek[q]$, respectively.  Since $\hat{\av{H}}$ is a non-singular projective curve the rational mapping $\mu\circ\zeta$ extends to the two to one morphism $\varphi:\hat{\av{H}}\rightarrow\pr{1}$. This mapping corresponds to some two-dimensional subspace of the Riemann space $L(P+Q)$ because it maps $P$ and $Q$ into a single point on $\pr{1}$.  Now, let there be another morphism  $\psi:\hat{\av{H}}\rightarrow\pr{1}$ of degree 2 which maps $P$ and $Q$ to the same point. Then $\psi$ is also determined be some linear subspace of $L(P+Q)$ of dimension two. However, by Riemann-Roch theorem, we have $\dim L(P+Q)=2$ and thus $\varphi$ and $\psi$ differs only by an automorphism of $\pr{1}$.

Since the mapping $\pi\circ\sigma^{-1}:\av{H}\dashrightarrow\pr{1}$ is a rational mapping of degree two which glues one pair of corresponding points together it must be a coherent mapping.

\end{proof}

Let $\vek[p]\in\av{X}$ be a regular point and let $f(x,y,z)=F(1,x,y,z)$ be the dehomogenization of the defining polynomial of $\av{X}$. Choose a plane given by the linear equation $h(x,y,z)=0$ passing through $\vek[p]$. If $\vek[q]\in\av{H}$ is a point corresponding to $\vek[p]$  then we know that a normal lines $N_{\vek[p]}\av{X}$ and $N_{\vek[q]}\av{X}$ intersect in a point $\vek[r]$ such that $|\vek[p\,r]|=|\vek[q\,r]|$. Hence we may write down the following system of equations for $\vek[q]$
\begin{equation}\label{eq corresponding point}
  \begin{array}{rcl}
    f(\vek[x]) &=&0\\
    h(\vek[x]) &=&0\\
    \det\left[\vek[p-x],\nabla\,f(\vek[p]),\nabla\,f(\vek[x])\right] &=& 0\\
    ((\vek[p-x])\cdot\nabla\,f(\vek[x]))^2|\nabla\,f(\vek[p])|^2-    ((\vek[p-x])\cdot\nabla\,f(\vek[p]))^2|\nabla\,f(\vek[x])|^2 &=& 0.
  \end{array}
\end{equation}
This system  has trivial solutions $\vek[p]$ and any surface singularity contained in the plane section $\av{H}$. If $\vek[q]$ is a non-trivial solution then satisfying the last two equations guarantees that the normal lines at $\vek[p]$ and $\vek[q]$ and the line joining these points are contained in a common plane and simultaneously the triangle which they form is isosceles.

For the surface given by a generic polynomial $f$ and a generic linear polynomial $h$, system \eqref{eq corresponding point} is not expected to have any non-trivial solution. On contrary if $f$ is the defining polynomial of a canal surface then the only expected non-trivial solution is exactly the corresponding point $\vek[q]$.

\begin{lem}
  Let $\av{X}$ be a rational canal surface with a regular point $\vek[p]\in\av{X}$. Then for a generic plane $h(\vek[x])=0$ passing through $\vek[p]$, the only non-trivial solution of~\eqref{eq corresponding point} is the~point corresponding to $\vek[p]$.
\end{lem}

\begin{proof}
  Let $(\vek[s](t),R(t))$ be a~parameterization of $\av{M}^2$, then $\av{X}$ is generated by the characteristic circles $\Sigma(t)=\Sigma'(t)=0$, c.f. \eqref{eq canal_eq}. We will derive the conditions under which  $\av{C}_0$, the characteristic circle $\Sigma(t_0)=\Sigma'(t_0)=0$, contains the point $\vek[q]$ corresponding to $\vek[p]$. Assume without loss of generality $\vek[s](t_0)\not\in N_{\vek[p]}\av{X}$. Then, since $N_{\vek[q]}\av{X}$ passes through $\vek[s](t_0)$ the intersection condition with $N_{\vek[p]}\av{X}$ can be reformulated as $\vek[q]$ has to be contained in the plane $\av{P}: (\nabla\,f(\vek[p])\times (\vek[p]-\vek[s](t_0))\cdot(\vek[x]-\vek[s](t_0))=0$.

If $\av{C}_0\subset\av{P}$ then one can easily check that there exists a point on the circle corresponding to $\vek[p]$. Nevertheless this condition is equivalent to $\Sigma'(t_0) = \av{P}$ which can be rewritten as $R'(t_0)=0$, $\vek[s]'(t_0)\cdot\nabla\,f(\vek[p])=0$ and $\vek[s]'(t_0)\cdot(\vek[p]-\vek[s]'(t_0))=0$ that can be fulfilled for finitely many values of $t$, only.  Hence, assume $\av{C}_0\not\subset\av{P}$ and denote by $\vek[q]_i(t_0)$ $i=1,2$ the two points from the intersection $\av{C}_0\cap\av{P}$. Computing $\vek[q]_i(t_0)$ consists in solving quadratic equations and thus coordinates of these points can be expressed using rational functions in $t$ and their square roots.
Thus the coordinates of $\vek[r]_i(t_0) = N_{\vek[p]}\av{X}\cap N_{\vek[q]_i}\av{X}$ can be written in the same way.  We conclude that $\vek[q]_i(t_0)$ is the point corresponding to $\vek[p]$ if and only if $|\vek[p]-\vek[r]_i(t_0)|^2=|\vek[q]_i(t_0)-\vek[r]_i(t_0)|^2$. Such a~system of equations can have only finitely many solutions. We conclude that except the points  lying on the same characteristic circles as $\vek[p]$ there exists only finitely many points corresponding to~$\vek[p]$ and thus a~generic plane misses them all.
\end{proof}

\subsection*{Case III: $\genus(\av{H})\geq 2$}

Whereas each curve of genus two is hyperelliptic, this does not hold for curves of higher genus anymore.  In addition, it is well known that a generic curve of genus at least three is not hyperelliptic.  Hence one first needs to test the hyperellipticity of $\av{H}$ and if the answer is positive then we can find its Weierstrass form. The following proposition, whose proof can be found in \cite[p. 117]{Sch94}, guarantees that the situation is much simpler compared to the elliptic case.

\begin{prp}\label{prp uniquenes of hyperellipticity}
  The two to one map from a hyperelliptic curve to $\pr{1}$ is unique up to an automorphism of $\pr{1}$
\end{prp}

\begin{cor}
  Let $\av{E}$ be Weierstrass form of $\av{H}$ as above and let $\pi$ be the projection mapping $(x,y)\in\av{E}$ to its $x$-coordinate. Then $\pi\circ\sigma^{-1}:\av{H}\rightarrow\pr{1}$ is a~coherent mapping.
\end{cor}

\begin{proof}
  A coherent mapping exists by Lemma~\ref{lem hyperelliptic projection}. Since $\pi\circ\sigma^{-1}$ is another mapping from $\av{H}$ to $\pr{1}$ of degree two, it differs from the~mentioned coherent mapping only by $\mathrm{Aut}(\pr{1})$, cf. Proposition~\ref{prp uniquenes of hyperellipticity}. And thus it is coherent as well.
\end{proof}

\section{Reconstruction of the squared MAT from the coherent mapping}\label{sec reconstruction}

Let $\mu:\hat{\av{H}}\rightarrow\av{H}$ be the~desingularization of $\av{H}$ and let $\vek[s](t)$ be a proper parameterization of the spine curve $\av{S}$. Then we have two mappings $\hat{\av{H}}\rightarrow\pr{1}$  of degree two, namely $\vek[s]^{-1}\circ\xi\circ\mu$ and $\pi\circ\sigma^{-1}\circ\mu$, where $\pi\circ\sigma^{-1}$ is the coherent mapping constructed in the previous section and $\xi$ is the~mapping from Lemma~\ref{lem hyperelliptic projection}.  Since these two mappings glue the pairs of corresponding points they have to differ only by an automorphism $\phi:\pr{1}\rightarrow\pr{1}$  (c.f. the proof of Theorem~\ref{thm construction of coherent map} and Proposition~\ref{prp uniquenes of hyperellipticity} for elliptic and hyperelliptic curves, respectively). The structure of above mappings is summarized in the diagram
\begin{equation}\label{eq rational param of MAT}
 \xymatrix{
      \hat{\av{H}}\ar[rd]^{\mu}\ar@/_/[rdd]\ar@/^2pc/[rrrd]    &&&\\
      &\av{H}\ar@{-->}[d]^{\pi\circ\sigma^{-1}}\ar@{-->}[r]^{\xi} &\av{S}\ar@{-->}[r]^{\vek[s]^{-1}} &\pr{1}\\
      &\pr{1}\ar[urr]_{\exists\phi\in\mathrm{Aut}(\pr{1})} &&
      }
\end{equation}

\bigskip
The mapping $\pi$ enabled us to parameterize the curve $\av{E}$ in the Weierstrass form by a square-root parameterization $\f e^{\pm}(t)$. Then via $\sigma$ we arrive at the square-root parameterization $\f h^{\pm}(t)$ of $\av{H}$ that correspond in parameter (which means that for any $t_0$ the points $\f h^+(t_0)$ and $\f h^-(t_0)$ are lying on the same characteristic circle of $\av{X}$). Intersecting the normals constructed at these corresponding points yields a parameterization of the rational spine curve $\av{S}$. This parameterization, cf. \eqref{eq rational param of MAT}, can be expressed as $\vek[s] \circ\phi$, which is a proper rational parameterization of $\av{S}$. The computation of the squared radius easily follows and we obtain a parameterization of $\av{M}^2$. Finally, using \cite{BaJuLaSchSi14} we compute from $\av{M}^2$ a parameterization $\f x(t,s)$ of the associated canal surface and if $f(\f x(t,s))=0$ we can conclude that $\av{X}$ is a canal surface with the squared medial axis transform $\av{M}^2$.

\medskip
The whole procedure is summarized in Diagram~\ref{diagram}. The methods and approaches studied in this paper will be now presented in detail in the following two particular examples.

\begin{exmp}\rm
Let $\av{X}_{\cis[R]}$ be an implicit surface in $\euR{3}$ given by the defining polynomial
\begin{equation}
   f(x,y,z) = x^3+xy^2+xz^2+2x^2+3y^2+z^2-5x-6.
\end{equation}
It can be shown that $\av{X}_{\cis[R]}$ is not a surface of revolution and its projective closure contains the absolute conic section with the multiplicity one, i.e., it holds $\mathrm{c}(\av{X})>0$.

A generic plane section is a non-singular cubic curve which possesses the genus one, i.e., it is an elliptic curve. For further computations, we find a regular point $\f p =\left(\frac{3}{2},0,\frac{3}{2}\right)$ on the canal surface and choose a plane which contains it, e.g. $h(x,y,z)=x-z+3=0$.  Next, we have to compute the point $\f q$  corresponding to $\f p$. The system \eqref{eq corresponding point} of the polynomial equations
\begin{equation}
\begin{array}{rcl}
x^3+xy^2+xz^2+2x^2+3y^2+z^2-5x-6 & = & 0 \\[0.8ex]
x-z+3 & = & 0 \\[0.8ex]
15x^2y+15y^3+15yz^2-45xy+80yz-120y & = & 0 \\[0.8ex]
-81x^6+216x^5z-290x^4y^2-209x^4z^2+\\
240x^3y^2z+240x^3z^3-209x^2y^4-328x^2y^2z^2-119x^2z^4+ & &\\
24xy^4z+ 48xy^2z^3+24xz^5+9y^4z^2+18y^2z^4+9z^6-297x^5+ & &\\
459x^4z-928x^3y^2-368x^3z^2+588x^2y^2z+ & &\\
204x^2z^3-759xy^4-630xy^2z^2-15xz^4-63y^4z-126y^2z^3-63z^5+ & &\\
396x^4-1473x^3z+593x^2y^2+974x^2z^2-267xy^2z-627xz^3-396y^4+ & &\\
342y^2z^2+50z^4+1782x^3-1803x^2z+2694xy^2+608xz^2-963y^2z+ & &\\
453z^3-621x^2+2595xz+1179y^2-849z^2-2385x-450z+1350
 &=& 0
\end{array}
\end{equation}
gives the only non-trivial solution $\f q =(1,0,2)$.

\medskip
The Weierstrass form $\av{E}$ of $\av{H}$ is
\begin{equation}
y_0^2+x_0^3-\frac{364}{3}x_0+\frac{13376}{27}=0
\end{equation}
and we arrive at the preimages of the corresponding points
\begin{equation}
\f p_0=\sigma^{-1}(\f p)=\left(\frac{22}{3},0\right), \quad \f q_0=\sigma^{-1}(\f q)=\left(\frac{16}{3},0\right).
\end{equation}
The third intersection point $\f r$ of $\av{E}$ with the line $\f p_0 \f q_0$ is $\left(-\frac{38}{3},0\right)$. The mapping $\sigma^{-1}$ composed with the projection from the point $\f r$ to the line $x=0$ yields the coherent mapping. The ``inverse'' of the projection is
given by the square-root parameterization of $\av{E}$ in the form
\begin{equation}
\f e^{\pm}(t)=
\begin{pmatrix}
{\frac {19}{3}}-{\frac {9}{2888}}\,{t}^{2}\pm{
\frac {1}{2888}}\,\sqrt {81\,{t}^{4}-987696\,{t}^{2}+8340544}\\[2ex]

\left( -{\frac {27}{109744}}\,{t}^{2}+\frac{3}{2}\pm{
\frac {3}{109744}}\,\sqrt {81\,{t}^{4}-987696\,{t}^{2}+8340544}
 \right) t
\end{pmatrix}\, .
\end{equation}
Hence, we arrive at the following square-root parameterization of $\av{H}$
\begin{equation}
\f h^{\pm}(t)=
\begin{pmatrix}
\displaystyle
-\frac{1}{2}\cdot \frac{ 126-{\frac {81}{1444}}\,{t}^{2
}\pm{\frac {9}{1444}}\,\sqrt {81\,{t}^{4}-987696\,{t}^{2}+8340544}
 }{ -51-{\frac {27}{2888}}\,{t}^{2}\pm{\frac {3}{2888}}\,
\sqrt {81\,{t}^{4}-987696\,{t}^{2}+8340544} }\\[4ex]

\displaystyle
-3t\cdot \frac{ -{\frac {27}{109744}}\,{t}^{2}+3/2\pm{
\frac {3}{109744}}\,\sqrt {81\,{t}^{4}-987696\,{t}^{2}+8340544}
 }{-51-{\frac {27}{2888}}\,{t}^{2}\pm{\frac {3}{2888}}\,
\sqrt {81\,{t}^{4}-987696\,{t}^{2}+8340544} }\\[4ex]

\displaystyle
3+\frac{1}{2}\cdot \frac{126-{\frac {81}{1444}}\,{t}^{2}\pm{
\frac {9}{1444}}\,\sqrt {81\,{t}^{4}-987696\,{t}^{2}+8340544} }{
 -51-{\frac {27}{2888}}\,{t}^{2}\pm{\frac {3}{2888}}\,\sqrt {81\,
{t}^{4}-987696\,{t}^{2}+8340544} }
\end{pmatrix}\, .
\end{equation}
From this we directly obtain a rational parameterization of $\av{M}^2$ in the form
\begin{equation}
\left(\frac{9}{5776}t^2-\frac{1}{2},-\frac{3}{38}t,0, \left(\frac{9}{5776}t^2+\frac{5}{2}\right)^2\right).
\end{equation}
In addition, in this case the standard medial axis transform $\av{M}$ is rational as well.
Finally, we compute the parameterization $\f x(s,t)$  of the associated canal surface associated to $\av{M}^2$. And as the condition $f(\f x(t,s))=0$ is satisfied we conclude that $\av{X}_{\cis[R]}$ is a rational canal surface with the squared medial axis transform~$\av{M}^2$.
\end{exmp}

\bigskip
\begin{exmp}\rm
Let $\av{X}_{\cis[R]}$ be an implicit surface in $\euR{3}$ given by the defining polynomial
\begin{equation}
\begin{array}{rcl}
f(x,y,z) & = & 256x^6+768x^4y^2+256x^4z^2+768x^2y^4+512x^2y^2z^2+256y^6+256y^4z^2-1536x^5+\\
& & +512x^4z-3072x^3y^2-1024x^3z^2-128x^2y^2z+512x^2z^3-1536xy^4-1024xy^2z^2-\\
& & -640y^4z-512y^2z^3+3712x^4 -2048x^3z+4928x^2y^2-1152x^2yz+1664x^2z^2+\\
& & +256xy^2z--1024xz^3+784y^4-1152y^3z+1024y^2z^2-1024yz^3+256z^4-\\
& & -4608x^3+576x^2y+2656x^2z-3712xy^2+2304xyz-1280xz^2-288y^3-448y^2z+\\
& & +384yz^2+128z^3+3232x^2-1152xy-1216xz+584y^2-1056yz+496z^2-\\
& & -1344x+344y+120z+257.
\end{array}
\end{equation}
Firstly, it can be shown that $\av{X}_{\cis[R]}$ is not a surface of revolution. In addition, its projective closure contains the absolute conic section with the multiplicity one and thus the surface has the non-zero cyclicity $\mathrm{c}(\av{X})$.

Since a generic plane section has the genus two it is a hyperelliptic curve. For further computations, we choose the plane $h(x,y,z)=x+2y-z=0$. The Weierstrass form of $\av{H}$ is
\begin{equation}
\begin{array}{rcl}
y_0^2+194693718016x_0^6-556701816832x_0^5+663351619840x_0^4- & &\\
-421618428672x_0^3+150753585024x_0^2-28751299584x_0+2284923704 & = &0.
\end{array}
\end{equation}
We compute $\f e^{\pm}(t)$, then $\f h^{\pm}(t)$ and from them we obtain the rational parameterization of $\av{M}^2$ in the form
\begin{equation}
\left(1, \frac {-7(2t-1)}{12t-5},{\frac {49(4{t}^{2}-4t+1)}{(12t-5)^{2}}},\frac {-7(2t-1)}{12t-5}\right).
\end{equation}
Finally, we compute the parameterization $\f x(s,t)$  of the associated canal surface associated to $\av{M}^2$. And as the condition $f(\f x(t,s))=0$ holds we conclude that $\av{X}_{\cis[R]}$ is a rational canal surface with the squared medial axis transform $\av{M}^2$.
\end{exmp}

\begin{table}[t]
\begin{tikzpicture}[node distance=2cm]
  \node (start) [start] {Input: $\av{X}:F(W,X,Y,Z)$};
  \node (dec1) [decision, below of=start, yshift=-0.3cm] {Is $\av{X}$ rotational?};
  \node (stop1) [stopyes, right of=dec1,xshift=5cm] {$\av{X}$ is a canal surfaces, see \cite{VrLa14} for more details};
  \node (dec2) [decision, below of=dec1, yshift=-0.3cm] {$\cyc(\av{X})>0$};
  \node (stop2) [stopno, right of=dec2, xshift=5cm] {$\av{X}$ is not a canal suraface by Proposition~\ref{prp non-zero cyclicity}};
  \node (dec3) [decision, below of=dec2, yshift=-0.3cm] {$\genus(\av{H})=0$};
  \node (stop3) [stopno, right of=dec3,xshift=5cm] {$\av{X}$ is not a canal surface by Theorem~\ref{thm rational section then sor}};

  \node (dec4) [decision, below of=dec3, yshift=-0.3cm] {$\genus(\av{H})=1$};

  \node (proc1) [process, text width=20em,right of=dec4, xshift=7cm] {\vek[p]:= regular point on $\av{X}$\\ \vek[q]:= corresponding point, see \eqref{eq corresponding point}};

  \node (dec6) [decision, below of=proc1,yshift=-0.25cm] {$\exists\vek[q]$?};
  \node (proc2) [process, below of=dec6,yshift=0.5cm] {$\sigma:\av{E\rightarrow H}$ Weierstrass form, see~\cite{vH95}};
  \node (proc3) [process, below of=proc2,yshift=0.5cm] {$\vek[r]:=$ intersection of  $\av{E}$ with line $\sigma^{-1}(\vek[p])\sigma^{-1}(\vek[q])$};
  \node (proc4) [process, below of=proc3,yshift=0.5cm] {$\pi:\av{E}\rightarrow\pr{1}$ projection from $\vek[r]$};

  \node (dec5) [decision, below of=dec4, yshift=-0.3cm] {Is $\av{H}$ hyperelliptic?};
  \node (stop4) [stopno, right of=dec5,xshift=2.6cm] {$\av{X}$ is not canal surfaces};
  \node (proc1a) [process, text width=20em,below of=dec5,yshift=-0.3cm] {$\sigma:\av{E\rightarrow H}$ Weierstrass form, see \cite{vH02}
                                           $\pi:\av{E}\rightarrow\pr{1}$ standard projection   };

  \node (stop5a) [stopyes, below of=proc1a,yshift=-1cm] {$\pi\circ\sigma^{-1}:\av{H}\dashrightarrow\pr{1}$};

  \draw [arrow] (start) -- (dec1);
  \draw [arrow] (dec1) -- node[anchor=south] {yes} (stop1);
  \draw [arrow] (dec1) -- node[anchor=west] {no} (dec2);
  \draw [arrow] (dec2) -- node[anchor=south] {no} (stop2);
  \draw [arrow] (dec2) -- node[anchor=west] {yes} (dec3);
  \draw [arrow] (dec3) -- node[anchor=south] {yes} (stop3);
  \draw [arrow] (dec3) -- node[anchor=west] {$\genus(\av{H})>0$} (dec4);
  \draw [arrow] (dec4) -- node[anchor=south] {yes} (proc1);
  \draw [arrow] (dec4) -- node[anchor=west] {$\genus(\av{H})\geq 2$} (dec5);
  \draw [arrow] (dec5) -- node[anchor=south] {no} (stop4);
  \draw [arrow] (dec5) -- node[anchor=west] {yes} (proc1a);

  \draw [arrow] (proc1) -- (dec6);
  \draw [arrow] (dec6) -- node[anchor=west] {yes} (proc2);
  \draw [arrow] (dec6) -- node[anchor=south] {no} (stop4);

  \draw [arrow] (proc2) -- (proc3);
  \draw [arrow] (proc3) -- (proc4);

  \draw [arrow] (proc4) |- (stop5a);
  \draw [arrow] (proc1a) -- (stop5a);

\end{tikzpicture}
\caption{Recognition of rational canal surfaces summarized.}\label{diagram}

\end{table}

%%%%%%%%%%%%%%%%%%%%%%%%%%%%%%%%%%%%%%%%%%%%%%%%%%%%%%%%%%%%%%%%%%%%%%%%%%%%%%%%%%%%%%%%%%%%%%%%%%%%%%%%%%%%%%%%%%%%%%%%%%%%%%%%%%%%%%%
\section{Conclusion}\label{Concl}
%%%%%%%%%%%%%%%%%%%%%%%%%%%%%%%%%%%%%%%%%%%%%%%%%%%%%%%%%%%%%%%%%%%%%%%%%%%%%%%%%%%%%%%%%%%%%%%%%%%%%%%%%%%%%%%%%%%%%%%%%%%%%%%%%%%%%%%

In this contribution we studied an interesting (and till now unsolved) theoretical problem, motivated by some technical applications (e.g. when implicit blends are constructed and then shall be characterized), i.e., how to recognize a rational canal surface from its defining polynomial equation. We designed an algorithm returning for rational canal surfaces also their squared medial axis transform which can be used for computing a rational parameterization of the given surface. The methods and approaches were presented on two commented examples. The investigation may be considered as another step towards the recognition of other implicitly given surfaces which play an important role in technical practice.

%%%%%%%%%%%%%%%%%%%%%%%%%%%%%%%%%%%%%%%%%%%%%%%%%%%%%%%%%%%%%%%%%%%%%%%%%%%%%%%%%%%%%%%%%%%%%%%%%%%%%%%%%%%%%%%%%%%%%%%%%%%%%%%%%%%%%%%
\section*{Acknowledgments}
%%%%%%%%%%%%%%%%%%%%%%%%%%%%%%%%%%%%%%%%%%%%%%%%%%%%%%%%%%%%%%%%%%%%%%%%%%%%%%%%%%%%%%%%%%%%%%%%%%%%%%%%%%%%%%%%%%%%%%%%%%%%%%%%%%%%%%%

The first author was supported by the project NEXLIZ, CZ.1.07/2.3.00/30.0038, which is co-financed by the European Social Fund and the state budget of the Czech Republic.
The work on this paper was supported by the European Regional Development Fund, project ``NTIS~--~New Technologies for the Information Society'', European Centre of Excellence, CZ.1.05/1.1.00/02.0090.
%We thank to all referees for their valuable comments, which helped us significantly to improve the paper.

\bibliographystyle{ieeetr}
\bibliography{bibliography}{}

%\bigskip
%\begin{appendix}
%
%\end{appendix}

\end{document}